\documentclass[11pt]{article}
\usepackage[margin=1.1in]{geometry}
\usepackage{amsmath,amssymb,amsthm,mathtools}
\usepackage{microtype}
\usepackage{booktabs,array}
\usepackage{enumitem}
\usepackage{needspace}
\usepackage{placeins}
\usepackage{tikz}
\usetikzlibrary{arrows.meta,positioning}
\usepackage[colorlinks=true,linkcolor=blue,citecolor=blue,urlcolor=blue,backref=page]{hyperref}
\usepackage[nameinlink,noabbrev]{cleveref}

\newtheorem{theorem}{Theorem}[section]
\newtheorem{lemma}[theorem]{Lemma}
\theoremstyle{definition}
\newtheorem{definition}[theorem]{Definition}

\newcommand{\Tow}{\operatorname{Tow}}
\newcommand{\E}{\mathbb{E}}
\newcommand{\mcalf}{\mathcal{F}}
\newcommand{\mcala}{\mathcal{A}}

\newcommand{\Tr}{{\textnormal{Tr}}}

\Crefname{theorem}{Theorem}{Theorems}
\Crefname{lemma}{Lemma}{Lemmas}
\Crefname{definition}{Definition}{Definitions}
\Crefname{remark}{Remark}{Remarks}
\Crefname{theorem}{Theorem}{Theorems}
\Crefname{lemma}{Lemma}{Lemmas}
\Crefname{definition}{Definition}{Definitions}
\Crefname{remark}{Remark}{Remarks}

\title{Should Tables Be Sorted? Revisited with a Large Language Model}
\author{Songhua He\\
    {\normalsize Rutgers University}\\
    {\small\texttt{songhuahe.cs@gmail.com}}}
\date{September 2026}

\begin{document}
\maketitle

\begin{abstract}
We revisit the implicit membership problem in Yao's full-table model
\cite{yao1981should} and obtain, to our knowledge, the first quantitative
improvements to his 45-year-old Ramsey bounds, most notably reducing the
two-probe bound from tower-type to polynomial.  In this model, a set
$S\subseteq\{1,\ldots,m\}$ of size $n$ is stored as a permutation of its
elements in an $n$-cell table, and a query algorithm probes cells to decide
whether $x\in S$.  Let $G_q(n)$ be the largest universe size $m$ for which a
$q$-probe membership scheme works for every size-$n$ subset of
$\{1,\dots,m\}$.  Yao proved that, for every fixed $n$, binary search for $x$
in a sorted table is optimal once $m$ is sufficiently large, but the Ramsey
argument gives an enormous tower-type upper bound on $G_q(n)$.  Yao explicitly
asked for the behavior of $G_q(n)$.  He determined the one-probe case exactly,
proving $G_1(n)=2n-2$ for $n>2$, but the behavior for $q\ge2$ remained wide
open.  On the construction side, Fiat and Naor \cite{fiat1993implicit} showed
that, for every large enough constant $q$, membership can be answered with $q$
probes for universe size up to $\exp(n^c)$ for some constant $c>0$.  For the
first adaptive case, $q=2$, we prove
$$
        G_2(n)=O(n^2(\log n)^2).
$$
For every fixed integer $q\ge3$, we show that $G_q(n)$ is at most a tower of
height $q-1$ with top $n^{1+o(1)}$; in particular,
$G_3(n)\le\exp(n^{1+o(1)})$.  The two-probe proof avoids Ramsey theory
altogether; for larger fixed $q$, we use Ramsey theory only to make the first
$q-1$ probes follow a fixed pattern, and then handle the last probe by the same
non-Ramsey argument.

Somewhat surprisingly, for each fixed $q$, we also show that implicit
membership is as hard as implicit search up to a polynomial loss in universe
size.  Implicit search is the harder variant that must reject if $x\notin S$
and otherwise output the cell containing $x$ \cite{fiat1993implicit}.  Let
$H_q(n)$ denote the analogous search threshold.  Trivially,
$H_q(n)\le G_q(n)$; conversely, for every $q,n$ we prove
$$
        G_q(n)\le n^q\bigl(H_q(n)+1\bigr)^{q+1}.
$$
Thus, for every fixed $q$, one threshold is at most $\exp(n^{O(1)})$ if and
only if the other is.  The proofs were first generated by ChatGPT 5.5 Pro
without mathematical hints; the membership-search equivalence emerged while
pursuing an improved four-probe bound.  The authors have validated and edited
the proofs and assume responsibility for all content.

\end{abstract}

\newpage
\tableofcontents
\newpage

\section{Introduction}

Dictionary lookup is one of the most fundamental tasks in data structures.
Standard dictionary models may use empty or additional cells, and cell
contents may encode arbitrary auxiliary information such as hash-function
descriptions.  Yao's full-table model isolates the opposite extreme
\cite{yao1981should}.  The data structure has exactly $n$
cells, each containing one distinct key from a set
$S\subseteq\{1,\ldots,m\}$ of size $n$.  Thus the storage scheme can choose
only the order of the keys in the table.  A query algorithm, knowing the
scheme, receives $x\in\{1,\ldots,m\}$ and probes cells adaptively until it
decides whether $x\in S$.

The question is not how many bits are needed to store a dictionary, but how
many probes are forced when the representation is exactly an ordering of the
stored keys.  Storing the keys in sorted order supports membership queries in
at most $\lceil\log_2(n+1)\rceil$ probes by binary search.  Yao proved that for
every fixed $n$, sorted order is optimal once the universe size $m$ is
sufficiently large.  Yao's proof was one of the earliest striking uses of
Ramsey theory to prove an impossibility result in computer science.
Quantitatively, however, the resulting bound is enormous.  Writing $G_q(n)$
for the largest universe size for which $q$ probes suffice, his proof gives the
upper bound\footnote{Throughout this work, upper bounds and lower bounds refer to
bounds on the threshold function $G_q(n)$.  Thus upper bounds are impossibility
results for $q$-probe schemes, whereas lower bounds are constructions.}
$$
        G_q(n)\le R_{n!}^{(n)}(2n-1)-1,
        \qquad q<\lceil\log_2(n+1)\rceil.
$$
Here $R_c^{(r)}(k)$ is the usual $c$-color Ramsey number for colorings of
$r$-subsets; the displayed bound has tower height $n$.
Motivated by the size of this bound, Yao explicitly asked to determine, for
given $n$ and $q$, the largest universe size for which $q$ probes suffice
\cite{yao1981should}.

He determined the one-probe threshold exactly: $G_1(n)=2n-2$ for $n>2$.
The natural next case is $q=2$.  It is the first setting in which adaptivity
can matter.  This leaves the first adaptive case:\footnote{The REGS program at the
University of Illinois explicitly singled out the two-probe case on a 2010
open-problem page and noted that even a quadratic-size construction was not
known; see \url{https://dwest.web.illinois.edu/regs/yao.html}.}
\begin{quote}
    Can two adaptive probes exploit the permutation of the keys to handle a
    universe much larger than linear, or does the full-table restriction
    keep two probes fundamentally weak?
\end{quote}
We prove the near-quadratic bound $G_2(n)=O(n^2(\log n)^2)$.  More generally,
for every fixed integer $q\ge3$, we reduce Yao's tower height from $n$ to
$q-1$.
In addition, we give a coarse equivalence between membership and the
implicit search problem studied by Fiat and Naor \cite{fiat1993implicit},
where a query must
reject if $x\notin S$ and, if $x\in S$, must output the address of the cell
containing $x$.
Writing $H_q(n)$ for the analogous search threshold, search is harder than
membership, so $H_q(n)\le G_q(n)$.  Our converse reduction incurs only a
polynomial loss in the universe size.  Consequently, for every fixed integer
$q$, $G_q(n)\le\exp(n^{O(1)})$ if and only if
$H_q(n)\le\exp(n^{O(1)})$.

\Needspace{10\baselineskip}
\subsection{Results}

Recall that $G_q(n)$ denotes the largest universe size on which a $q$-probe
implicit membership scheme works for all $n$-subsets in Yao's full-table
model.

\begin{theorem}[Two-probe upper bound]\label{thm:intro-g2}
$$
        G_2(n)=O(n^2(\log n)^2).
$$
\end{theorem}

The proof is entirely non-Ramsey.  It replaces Yao's tower-type upper bound for
the two-probe case by a near-quadratic one.

The next theorem improves Yao's bound for every fixed integer $q\ge3$.  Here
$\Tow_t(x)$ denotes a tower of $2$'s of height $t$ with top exponent $x$, so
$\Tow_1(x)=x$ and $\Tow_{t+1}(x)=2^{\Tow_t(x)}$.

\begin{theorem}[Fixed-$q$ upper bound]\label{thm:intro-fixedq}
For every fixed integer $q\ge3$,
$$
        G_q(n)\le \Tow_{q-1}(n^{1+o(1)}).
$$
\end{theorem}

Thus the tower height drops from $n$ in Yao's original Ramsey argument to
$q-1$.  In particular,
$$
        G_3(n)\le \exp(n^{1+o(1)}).
$$

\begin{table}[!htbp]
\centering
\small
\setlength{\tabcolsep}{1.5pt}
\begin{tabular}{@{}>{\raggedright\arraybackslash}m{0.18\linewidth}
                    >{\raggedright\arraybackslash}m{0.34\linewidth}
                    >{\raggedright\arraybackslash}m{0.26\linewidth}
                    >{\raggedright\arraybackslash}m{0.20\linewidth}@{}}
\toprule
Probe budget &
Known lower bound &
Previous upper bound &
Upper bound from this work \\
\midrule
$q=1$, $n>2$ &
$G_1(n)\ge2n-2$ \cite{yao1981should} &
$G_1(n)\le2n-2$ \cite{yao1981should} &
-- \\
\midrule
$q=2$ &
$G_2(n)\ge3n-4$ \cite{jobson2011avoiding} &
$\Tow_n(n^{O(n)})$ \cite{yao1981should} &
$O(n^2(\log n)^2)$ \\
\midrule
fixed $q\ge3$ &
$G_q(n)\ge3n-4$ \cite{jobson2011avoiding} &
$\Tow_n(n^{O(n)})$ \cite{yao1981should} &
$\Tow_{q-1}(n^{1+o(1)})$ \\
\midrule
sufficiently large fixed $q$ &
\begin{tabular}{@{}l@{}}
$G_q(n)\ge n^{c_q}$ constructive\\
$G_q(n)\ge \exp(n^{c_q'})$ nonconstructive
\end{tabular}
\cite{fiat1993implicit} &
$\Tow_n(n^{O(n)})$ \cite{yao1981should} &
$\Tow_{q-1}(n^{1+o(1)})$ \\
\midrule
$q\ge\lceil\log_2(n+1)\rceil$ &
$G_q(n)=\infty$ &
-- &
-- \\
\bottomrule
\end{tabular}
\caption{A comparison of known bounds on $G_q(n)$.  Lower bounds are
constructions in Yao's full-table model; upper bounds are impossibility results
for universes above the displayed scale.  In the fourth row, $c_q,c_q'>0$
depend only on $q$.  The last row marks the point at which sorted order gives
binary search and the threshold becomes infinite.}
\label{tab:gq-bounds}
\end{table}
\FloatBarrier

Our last result is a structural connection between implicit membership and its
search variant.  Recall that $H_q(n)$ is the corresponding universe threshold
for implicit search.  Since search is harder than membership,
$H_q(n)\le G_q(n)$.  The next theorem gives a converse comparison with only a
polynomial loss for each fixed integer $q$.

\begin{theorem}[Extracting search from membership]
\label{thm:intro-membership-search}
For every $q,n$,
$$
        G_q(n)\le n^q\bigl(H_q(n)+1\bigr)^{q+1}.
$$
\end{theorem}

Together with $H_q(n)\le G_q(n)$, this gives a coarse equivalence: for every
fixed integer $q$, $G_q(n)\le\exp(n^{O(1)})$ if and only if
$H_q(n)\le\exp(n^{O(1)})$.

Fiat and Naor explicitly identified the gap between exponential constructions
and constant-height-tower impossibility bounds for implicit search as an open
problem \cite{fiat1993implicit}.  Since $H_q(n)\le G_q(n)$,
\Cref{thm:intro-g2,thm:intro-fixedq} improve their $q$-probe search upper bound
from tower height
$q+2\lceil\log_2 q\rceil$ to a polynomial for $q=2$ and to tower height $q-1$
for every fixed $q\ge3$.

\subsection{Related work}

\paragraph{Previous work on implicit membership.}
Yao introduced the full-table membership problem and proved the exact one-probe
threshold $G_1(n)=2n-2$ for $n>2$ \cite{yao1981should}.  For larger constant
probe bounds, his general Ramsey argument gives a finite upper bound of tower
height $n$.  On the lower-bound side, Jobson gives
an adaptive two-probe construction with universe size $3n-4$
\cite{jobson2011avoiding}.  Howard studied the nonadaptive two-probe
restriction, in which the two probes are simultaneous: both cell addresses are
fixed by the query before either cell value is seen
\cite{howard2014determining}.  In that model
he constructed schemes for $m\le \frac52 n-3$ when $n\ge4$, and proved that no
scheme exists for $n\ge30$ once
$m\ge 6\binom{n}{2}\log n$; the constant $6$ can be replaced by
$4+\varepsilon$ for all sufficiently large $n$.  Thus Howard obtained a
polynomial impossibility bound in the nonadaptive two-probe model, whereas no
polynomial upper bound was previously known for adaptive two-probe membership.

\paragraph{Previous work on implicit search.}
The study of implicit search was initiated by Fiat and Naor
\cite{fiat1993implicit}.  For every sufficiently large fixed integer $q$, they
give a constant $c_q>0$ such that
$H_q(n)\ge\exp(n^{c_q})$.  They also
introduced a combinatorial structure called a \emph{rainbow}, which coarsely
characterizes implicit search.  In short, a rainbow is an $n$-coloring of all
sequences of some fixed constant length from the universe $\{1,\dots,m\}$ such
that, for every $n$-subset $S$ of the universe, all $n$ colors occur among the
sequences whose entries lie in $S$.  They show that, when $m\le 2^n$,
constant-probe implicit search schemes exist if and only if rainbows exist.
On the impossibility side, they
show that a $q$-probe search scheme yields a rainbow on sequences of length
$q+2\lceil\log_2 q\rceil$.  The corresponding Ramsey bound
gives $H_q(n)\le
\Tow_{q+2\lceil\log_2 q\rceil}(n^{O_q(1)})$.  They also point out that
disperser constructions \cite{sipser1988expanders,cohen1989dispersers} imply
rainbow constructions.

Our \Cref{thm:intro-membership-search} adds the missing membership-to-search
direction by extracting search schemes from membership schemes, up to
polynomial loss.  See the illustration below.

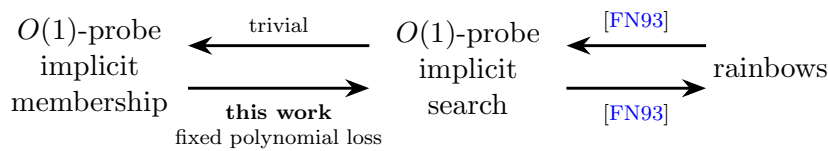
\begin{figure}[!htbp]
\centering
\vspace{0.35\baselineskip}
\begin{tikzpicture}[
    obj/.style={
        align=center,
        inner sep=0pt,
        text width=2.45cm
    },
    smallobj/.style={
        align=center,
        inner sep=0pt,
        text width=1.55cm
    },
    arr/.style={
        -{Stealth[length=3.2mm,width=2.4mm]},
        line width=1.1pt,
        shorten <=2pt,
        shorten >=2pt
    },
    lab/.style={font=\scriptsize, align=center, inner sep=0.5pt},
    >=Stealth
]
\node[obj] (membership) at (0,0)
    {$O(1)$-probe\\implicit membership};
\node[obj] (search) at (5.0,0)
    {$O(1)$-probe\\implicit search};
\node[smallobj] (rainbow) at (9.0,0)
    {rainbows};

\draw[arr] ([yshift=-8pt]membership.east) --
    node[lab, below=5pt] {\textbf{this work}\\fixed polynomial loss}
    ([yshift=-8pt]search.west);
\draw[arr] ([yshift=8pt]search.west) --
    node[lab, above=4pt] {trivial} ([yshift=8pt]membership.east);

\draw[arr] ([yshift=-8pt]search.east) --
    node[lab, below=5pt] {\cite{fiat1993implicit}} ([yshift=-8pt]rainbow.west);
\draw[arr] ([yshift=8pt]rainbow.west) --
    node[lab, above=4pt] {\cite{fiat1993implicit}} ([yshift=8pt]search.east);
\end{tikzpicture}
\caption{Relationships among implicit membership, implicit search, and
rainbows.  Arrows denote construction or existence implications.}
\label{fig:search-membership-rainbow}
\end{figure}

\paragraph{Other related dictionary problems.}
Fich and Miltersen studied the same static membership problem on a unit-cost
RAM and proved time lower bounds matching binary search on a sorted table
\cite{fich1995tables}.  Their model allows arbitrary integer encodings but
charges RAM computation, whereas Yao's full-table model gives computation for
free but restricts storage to a permutation of the stored keys.

A large body of dictionary and membership work allows the data structure to
store an arbitrary encoding of the set
\cite{fredman1984storing,buhrman2000bitvectors,pagh2001low,
raman2007succinct,patrascu2008succincter,garg2015set,garg2017set,
larsen2024optimalnonadaptive,hu2025optimalstatic}.  Yao's full-table model
removes this freedom: the table has exactly $n$ cells, each containing one
distinct stored key, so the only auxiliary information is the order of those
keys.  Thus our bounds do not lower-bound ordinary dictionaries; they measure
the cost of forbidding all encoded memory beyond the permutation itself.

Finally, there are broad literatures on implicit data structures,
and on implicit graph representations.  In implicit data structures,
there are related works on dynamic dictionaries that support insertions and
deletions
\cite{munro1980implicit,munro1986implicit,borodin1988tradeoff,
franceschini2006implicit,brodal2012cache}.  The goal there is to support
efficient updates and queries, whereas the implicit membership and search
problems studied in this work are static.  Implicit graph representations ask a
different question: how to label vertices of a graph so that adjacency can be
determined from the labels alone
\cite{kannan1992implicit,spinrad2003efficient,hatami2022implicit}.  Graph
representations are local decoding from chosen labels, whereas implicit
membership has no chosen labels, only the permutation of the stored keys.

\subsection{Proof ideas}

The common principle behind our improvements is to exploit the limited
information exposed by a constant-probe query.  Yao instead colors each
$n$-element subset by its entire table order and then applies Ramsey's theorem
to this coloring
\cite{yao1981should}.  On a large homogeneous set, every stored set has the
same relative rank order, reducing the problem to search in a fixed
permutation of sorted order.  This is conceptually clean, but it is too global
for a constant-probe algorithm: the color records the entire table permutation
although each query reads only a few cells.  The resulting Ramsey bound
therefore has tower height $n$, even when the query algorithm is allowed only
two probes.

\paragraph{Two probes.}
For adaptive two-probe schemes, the best previous upper bound was still the one
given by Yao's Ramsey argument.
The one-probe case was settled exactly by Yao, but his proof is a special
case analysis and gives little guidance once the second probe may depend on the
first value read.  On the other hand, small constructions show that two probes
can already beat one probe: Jobson shows that
$G_2(n)\ge 3n-4$ \cite{jobson2011avoiding}.  Thus the first real question is
whether adaptivity can make two probes powerful on superpolynomial universes.

Our answer is no.  The proof begins with a transcript-rectangle observation:
if two query-input pairs $(x,S)$ and $(x',S')$ have the same two probed cells
and the same two values read, then the crossed pairs $(x,S')$ and $(x',S)$ have
the same transcript as well.  Crossing the stored sets therefore preserves the
transcript for each fixed query.  The algorithm gives the same answer to $x$
on $S$ and $S'$, so correctness gives $x\in S$ if and only if $x\in S'$;
similarly, $x'\in S$ if and only if $x'\in S'$.

We next use the first probe to create a residual problem.  Since the first
address depends only on the query, there is a large set $Q$ of queries that
first probe a common cell.  Restrict to stored sets inside $Q$.  For each possible
key $y\in Q$ read from that first cell, let $\mcalf_y$ be the family of
residual sets $R$ of size $n-1$ for which the table of $R\cup\{y\}$ places
$y$ in the common first cell.  Every stored $n$-set inside $Q$ has a unique key
$y$ in that cell, and hence a unique representation $R\cup\{y\}$ with
$R\in\mcalf_y$.

Now fix $y$.  After the first probe has returned $y$, the second address on a
query $x\in Q\setminus\{y\}$ is already determined, so the queries split
into at most $n$ parts according to the second cell they inspect.  The key
observation is that, on each part $P_j$ consisting of the keys that inspect
cell $j$, the number of different intersections $R\cap P_j$ is at most $|Q|$.
Indeed, if $R,R'\in\mcalf_y$ and the tables for $R\cup\{y\}$ and
$R'\cup\{y\}$ have the same content in cell $j$, then for every $x\in P_j$ the
query-input pairs $(x,R\cup\{y\})$ and $(x,R'\cup\{y\})$ have the same
transcript.  Correctness therefore forces $x\in R$ if and only if $x\in R'$,
so $R\cap P_j=R'\cap P_j$.

The final step is a counting lemma.  If $|Q|$ were a sufficiently large
multiple of $n(\log n)^2$, the intersection bound above would make every fiber
$\mcalf_y$ have density at most $n^{-3}$ inside
$\binom{Q\setminus\{y\}}{n-1}$.  Summing over $y$ would then account for at most
$n^{-2}\binom{|Q|}{n}$ pairs $(y,R)$, fewer than the $\binom{|Q|}{n}$ pairs
needed to represent every $n$-subset of $Q$ exactly once.  Thus
$|Q|=O(n(\log n)^2)$, and since $|Q|\ge m/n$, we obtain
$$
        G_2(n)=O(n^2(\log n)^2).
$$

\paragraph{A constant number of probes.}
\label{pg:fixed-q_overview}
For a fixed integer $q\ge3$, the benchmark is again the upper bound from Yao's Ramsey
argument, which gives a
tower-type upper bound of height $n$.  The difficulty is that the two-probe proof cannot simply be iterated.
After several adaptive probes, the next address may depend on the whole sequence
of keys already seen, so the clean residual partition used in the two-probe
proof need not persist.  Fiat and Naor's implicit search scheme constructions make this warning
concrete: carefully arranged permutations can simulate
surprisingly rich auxiliary information \cite{fiat1993implicit}.

A first refinement already improves Yao's argument substantially.  For each
$(q+1)$-subset, color it by the list of at most
$q(q+1)!$ cell addresses probed in all formal executions obtained by ordering
its elements as the query and the possible replies.  There are at most
$n^{q(q+1)!}$ colors.  If
$$
        m\ge R_{n^{q(q+1)!}}^{(q+1)}(2n),
$$
Ramsey theory gives a homogeneous set $U$ of size $2n$.  On $U$, one fixed list
of at most $q(q+1)!$ cells contains every cell that can be probed by a query
whose visible values stay inside $U$, until the query has either seen $x\in U$ or
exhausted its $q$ probes.  By the rectangular property, if two stored sets
$S,S'\subseteq U$ have the same contents in these cells, then every query
$x\in U$ sees the same transcript and receives the same answer for $S$ and
$S'$.  Correctness therefore forces $x\in S$ iff $x\in S'$ for every
$x\in U$, so $S=S'$.  Hence the $\binom{2n}{n}$ possible stored sets in $U$
must induce distinct contents on these at most $q(q+1)!$ cells, but there are
only $(2n)^{q(q+1)!}$ such contents.  This is impossible because
$\binom{2n}{n}=2^{\Theta(n)}\gg (2n)^{q(q+1)!}$.  Hence this crude
probe-pattern Ramsey argument already gives
$$
        G_q(n)<R_{n^{q(q+1)!}}^{(q+1)}(2n)\le \Tow_{q+1}(n^{O_q(1)}),
$$
a tower whose height is only $q+1$, rather than $n$.

Our proof gains two Ramsey levels over this refinement.  We first pigeonhole a
common first cell, so the first probe requires no Ramsey step.  We then apply
Ramsey theory only to the intermediate probes $2,\ldots,q-1$ and leave the last
probe to the non-Ramsey counting argument, removing another Ramsey level.  To
obtain the sharper top of the tower, we partition the table into coarse blocks
and record only which blocks are reached, rather than the exact cell indices.
Together, these refinements yield an upper bound of tower height $q-1$ with top
$n^{1+o(1)}$.

\paragraph{Extracting search from membership.}
Our next result is a somewhat surprising equivalence between implicit membership and implicit search.
Since search is harder than
membership, any upper bound on $G_q(n)$ gives the same upper bound for
$H_q(n)$, so our
membership upper bounds rule out search at the same probe budget.  The
converse is not obvious: a membership algorithm may
accept a present key without ever seeing it, while a search algorithm must find
the cell that contains the key.

The proof is a probabilistic argument.
One would like to choose a subuniverse of $[m]$ that avoids every bad pair $(x,\tau)$ consisting of
a query $x$ and its accepting transcript $\tau$ that does not read $x$.  A direct
union bound is too weak, however, because it pays an extra factor of $m$ for the
choice of the query.  Indeed, after padding the algorithm to exactly $q$
probes, there are about $m\cdot n^qm^q$ possible query-transcript pairs, and
for each pair the values read together with $x$ form a set of $q+1$ elements
that would have to lie in the subuniverse; a random subuniverse of size $\ell$
would therefore leave, in expectation, survivors on the order of
$$
        m\cdot n^qm^q\left(\frac{\ell}{m}\right)^{q+1}
        =
        n^q\ell^{q+1},
$$
which is not small enough.\footnote{A uniformly random $\ell$-subset of
$\{1,\ldots,m\}$ contains any fixed set of $q+1$ elements with probability
about $(\ell/m)^{q+1}$.}  The point is that these bad pairs have a rectangular
structure.  Group them by the cells probed and the values read.  If a stored
set realizes one such transcript, then correctness forces it to contain every
query $x$ that the scheme accepts blindly with this transcript; otherwise the
same transcript would give a false positive.  Thus it suffices to choose one
representative blindly accepted query for each transcript, and to avoid only
the small forced set formed by this representative together with the values
read.

There are not many such forced sets, since a $q$-probe computation reads only
$q$ addresses and $q$ values.  After taking representatives, the same union
bound becomes
$$
        n^qm^q\left(\frac{\ell}{m}\right)^{q+1}
        =
        \frac{n^q\ell^{q+1}}{m},
$$
which is smaller than $1$ once
$m\gg n^q\ell^{q+1}$.  Since the number of surviving forced sets is
integer-valued, some subuniverse of size $\ell=H_q(n)+1$ avoids all
representative forced sets.  Then no bad pair can occur inside the restricted
universe.  Hence, on every
accepting query $x$, the membership scheme must actually read $x$, and it can
be turned into a search scheme by returning the cell where $x$ was found.

\section{Preliminaries}

\paragraph{Yao's full-table model.}
The universe is $[m]=\{1,\ldots,m\}$.  A storage scheme maps every set
$S\in\binom{[m]}{n}$ to an ordered table
$$
        T_S=(T_S[1],\ldots,T_S[n])
$$
whose entries are exactly the elements of $S$, each appearing once.  A
deterministic adaptive query algorithm receives $x\in[m]$ and probes cells of
$T_S$ one at a time.  The first probe address is a function of $x$ alone; later
addresses may depend on $x$ and on the cell contents already seen.
The storage map and all computation performed by the query algorithm are
unrestricted; only cell probes are counted.

For the membership problem, the query algorithm outputs a bit and must accept
exactly when $x\in S$.  A storage scheme together with such a query algorithm
is an \emph{implicit membership scheme} in Yao's full-table model.

For the search problem, the query algorithm outputs either ``not present'' or a
cell address.  It must output ``not present'' when $x\notin S$, and when
$x\in S$ it must output an address $i\in[n]$ with $T_S[i]=x$.  A storage
scheme together with such a query algorithm is an \emph{implicit search
scheme}.  Notice that a search scheme can be made to read $x$ before returning
its address with at most one additional probe: before reporting an address, it
probes that cell if it has not already done so.  Correctness ensures that
whenever an address is reported, that cell contains $x$.

We freely convert an ``at most $q$ probes'' algorithm into an ``exactly $q$
probes'' algorithm by adding dummy probes.  We also assume that an algorithm
does not reprobe a cell during a single query, since repeated probes reveal no
new information; this convention only simplifies the counting of transcript
patterns.

\paragraph{Universe thresholds.}
Let $f(n,m)$ be the minimum worst-case number of adaptive probes needed to solve
implicit membership for all $n$-subsets of $[m]$ in this model.  Define
$$
        G_q(n)=\max\{m:f(n,m)\le q\}.
$$
Thus $G_q(n)$ is the largest universe size for which $q$ probes suffice.  Yao's
one-probe theorem gives $G_1(n)=2n-2$ for $n>2$.

Let $h(n,m)$ be the analogous minimum worst-case number of probes for implicit
search, and define
$$
        H_q(n)=\max\{m:h(n,m)\le q\}.
$$
In both definitions, the maximum is understood to be $\infty$ if the
corresponding set of admissible universe sizes is unbounded.
A search scheme answers membership if we interpret every output other than
``not present'' as acceptance, so
$$
        H_q(n)\le G_q(n).
$$
Sorting the table and applying binary search solves both problems within
$\lceil\log_2(n+1)\rceil$ probes, so
$$
        G_q(n)=H_q(n)=\infty
$$
once $q\ge\lceil\log_2(n+1)\rceil$.  The interesting regime is therefore
$q<\lceil\log_2(n+1)\rceil$.

\paragraph{Asymptotic conventions.}
All asymptotics are as $n\to\infty$.  Whenever $q$ is fixed, it is fixed
independently of $n$, and implicit constants may depend on $q$ unless stated
otherwise.  All logarithms without an explicit base are base $2$; $\ln$
denotes the natural logarithm when it appears.  We write
$$
        \Tow_1(x)=x,
        \qquad
        \Tow_{t+1}(x)=2^{\Tow_t(x)}.
$$

\paragraph{Ramsey notation.}
For integers $r,c,k\ge1$, let $R_c^{(r)}(k)$ be the least $n_0$ such that every
coloring
$$
        \chi:\binom{[n_0]}{r}\to[c]
$$
contains a set $H\subseteq[n_0]$ of size $k$ on which $\chi$ is constant.  That
is, all $r$-subsets of $H$ receive the same color.

\begin{lemma}[Fixed-uniformity Ramsey estimate]
\label{lem:fixed_uniformity_ramsey}
For every fixed integer $r\ge2$, there is a constant $c_r$ such that, for all
$\lambda\ge2$ and $k\ge3$,
$$
        R_\lambda^{(r)}(k)
        \le
        \Tow_r(c_rk\lambda\log\lambda).
$$
\end{lemma}

\begin{proof}
    We proceed by induction on $r$, proving the claim simultaneously for all
    $\lambda\ge2$ and $k\ge3$.  For $r=2$, the standard multicolor graph Ramsey
    bound gives, for a suitable absolute constant $c_2$,
    $$
        R_\lambda^{(2)}(k)
        \le
        2^{c_2k\lambda\log\lambda}
        =
        \Tow_2(c_2k\lambda\log\lambda).
    $$

    Now fix $r\ge3$ and assume the claim for $r-1$.  If $k\le r$, every $k$-set
    contains at most one $r$-subset, so $R_\lambda^{(r)}(k)=k$ and the desired
    bound is immediate.  We may therefore assume that $k\ge r+1$.  The standard
    recurrence from
    \cite{erdos1952combinatorial} gives
    $$
        R_\lambda^{(r)}(k)
        \le
        2^{c\lambda\left(R_\lambda^{(r-1)}(k-1)\right)^{r-1}}
    $$
    for an absolute constant $c$.  Since $k-1\ge3$, the induction hypothesis
    applies.  Put
    $$
        a=c_{r-1}(k-1)\lambda\log\lambda.
    $$
    Then
    $$
        R_\lambda^{(r-1)}(k-1)\le \Tow_{r-1}(a).
    $$
    Taking $c_{r-1}\ge1$, we have $a\ge\lambda$.  Since $r$ is fixed, there is
    a constant $d_r$ such that
    $$
        c\lambda\bigl(\Tow_{r-1}(a)\bigr)^{r-1}
        \le
        \Tow_{r-1}(d_ra).
    $$
    The recurrence therefore yields
    $$
        R_\lambda^{(r)}(k)
        \le
        2^{\Tow_{r-1}(d_ra)}
        =
        \Tow_r(d_ra).
    $$
    Finally, $a\le c_{r-1}k\lambda\log\lambda$.  Choosing
    $c_r\ge\max\{1,d_rc_{r-1}\}$ proves the claim and completes the induction.
\end{proof}

\section{The two-probe upper bound}
\label{sec:two-probe}

Our main result in this section is a near-quadratic upper bound on $G_2(n)$.
Before this work, no exponential upper bound, let alone a polynomial upper
bound, was known for $G_q(n)$ for any fixed $q\ge 2$. The case $q=1$ was
settled by Yao~\cite{yao1981should}, but Yao's proof proceeds by a case
analysis that does not appear to extend to two probes. Thus $q=2$ is the first
case in which one must confront adaptivity.

We now prove \Cref{thm:intro-g2}.
The starting point is a rectangle argument. Fix a deterministic scheme. If two
query-input pairs $(x,S)$ and $(x',S')$ produce the same probe transcript, then
the crossed pairs $(x,S')$ and $(x',S)$ produce the same transcript as well.
Thus each transcript class has a rectangular structure in the query-input
incidence matrix. Correctness also implies that each such rectangle is
monochromatic in every query row: the same transcript on $(x,S)$ and $(x,S')$
forces $x\in S$ if and only if $x\in S'$.  We next restrict to a large set of
queries sharing the first cell and group stored sets by the key occupying that
cell.  Within each group, queries are partitioned by their second-probe cell,
and the row-wise consistency above makes the intersection with each resulting
query class determined by the content returned from that cell.  A counting
lemma then shows that, if $m$ exceeds a sufficiently large constant times
$n^2(\log n)^2$, the groups together account for fewer stored sets than the
number of $n$-subsets the scheme must represent.

\begin{definition}[Probe transcript]
    Fix a query algorithm. Given a query $x$ and a table $T_S$, the probe
    transcript is the sequence
    $$
        \tau(x,S)=((\alpha_1(x,S),\beta_1(x,S)),\ldots,
        (\alpha_q(x,S),\beta_q(x,S))),
    $$
    where $\alpha_i(x,S)\in[n]$ is the cell probed at step $i$, and
    $$
        \beta_i(x,S)=T_S[\alpha_i(x,S)]\in[m]
    $$
    is the value read from that cell.
\end{definition}

\begin{lemma}[Rectangle argument]
    \label{lem:rectangle}
    Fix a deterministic implicit membership scheme of probe complexity $q$
    that answers membership queries correctly. Fix a probe transcript
    $$
        \tau=((\alpha_1,\beta_1),\ldots,(\alpha_q,\beta_q)).
    $$
    Suppose that the scheme produces the transcript $\tau$ on both query-input
    pairs $(x,S)$ and $(x',S')$. Then the scheme also produces the transcript
    $\tau$ on both crossed pairs $(x,S')$ and $(x',S)$.

    Consequently,
    $$
        x\in S \Longleftrightarrow x\in S',
        \qquad
        x'\in S \Longleftrightarrow x'\in S'.
    $$
\end{lemma}

\begin{proof}
    Since the transcript on both $(x,S)$ and $(x',S')$ is $\tau$, the two
    tables $T_S$ and $T_{S'}$ contain the same value $\beta_i$ in cell
    $\alpha_i$ for every $i=1,\ldots,q$.

    We first show that the transcript on $(x,S')$ is also $\tau$. On query $x$,
    the algorithm probes cell $\alpha_1$, because this is what it does on the
    query-input pair $(x,S)$. The table $T_{S'}$ contains the same value
    $\beta_1$ in this cell as $T_S$ does. Therefore, after the first probe, the
    algorithm is in exactly the same state as it was on the query-input pair
    $(x,S)$. Inductively, after the first $i-1$ probes, the algorithm has seen
    the same query $x$ and the same values $\beta_1,\ldots,\beta_{i-1}$, so it
    probes the same next cell $\alpha_i$ and reads the same value $\beta_i$
    from $T_{S'}$. Thus the transcript on $(x,S')$ is $\tau$.

    Since the query $x$ and the transcript are the same on $(x,S)$ and
    $(x,S')$, the deterministic algorithm gives the same answer on these two
    inputs. Correctness gives
    $$
        x\in S \Longleftrightarrow x\in S'.
    $$
    The same argument, with $x'$ in place of $x$, shows that the transcript on
    $(x',S)$ is also $\tau$, and correctness gives
    $$
        x'\in S' \Longleftrightarrow x'\in S.
    $$
\end{proof}

\paragraph{Trace counting.}

Assume a two-probe scheme exists for universe size $m$. Since the first probe
depends only on the query, the pigeonhole principle gives a cell
$\alpha_1^*\in[n]$ and a set $Q\subseteq[m]$ with $|Q|\ge m/n$ such that every
query $x\in Q$ first probes $\alpha_1^*$. Write $m'=|Q|$, and restrict attention
to stored sets $S\in\binom{Q}{n}$.

For each $y\in Q$, define the \emph{first-cell fiber}
$$
    \mcalf_y
    :=
    \{R\subseteq Q\setminus\{y\}: |R|=n-1
    \text{ and } T_{R\cup\{y\}}[\alpha_1^*]=y\}.
$$
Equivalently, $R\in\mcalf_y$ means that the chosen first cell contains $y$ when
the stored set is $R\cup\{y\}$.  Every $S\in\binom{Q}{n}$ determines the unique
pair $y=T_S[\alpha_1^*]$ and $R=S\setminus\{y\}\in\mcalf_y$.

Conversely, every pair $(y,R)$ with $y\in Q$ and $R\in\mcalf_y$ yields the $n$-set
$S=R\cup\{y\}$, and these two operations are inverse.  Consequently,
\begin{equation}
    \label{eq:partition_f}
    \sum_{y\in Q}|\mcalf_y|=\binom{m'}{n}.
\end{equation}

Fix $y\in Q$. After the first probe has revealed $y$, the second probe on a
query $x\in Q\setminus\{y\}$ is fixed by the algorithm. This gives a partition
of $Q\setminus\{y\}$ into at most $n$ parts
$$
    \begin{aligned}
    P_j(y)
    :=
    \{x\in Q\setminus\{y\}:&
    \text{ after the first probe returns }y,\\
    &\text{the second probe on }x\text{ is cell }j\},
    \qquad j\in[n].
    \end{aligned}
$$
For $R\in\mcalf_y$, the \emph{second-probe trace} of $R$ on $P_j(y)$ is
$R\cap P_j(y)$.

We use the row-wise consequence of \Cref{lem:rectangle}: for a fixed query,
identical transcripts on two stored sets force the same membership answer.

\begin{lemma}[Second-probe trace constraint]
    \label{lem:second_trace}
    For every $y\in Q$ and every cell $j\in[n]$,
    $$
        \left|\{R\cap P_j(y):R\in\mcalf_y\}\right|\le m'.
    $$
\end{lemma}

\begin{proof}
    Suppose $R,R'\in\mcalf_y$ and the two corresponding tables have the same
    content in cell $j$:
    $$
        T_{R\cup\{y\}}[j]=T_{R'\cup\{y\}}[j].
    $$
    For every query $x\in P_j(y)$, both tables give the same two-probe
    transcript: the first cell is $\alpha_1^*$ with value $y$, and the second
    cell is $j$ with the same value in both tables. By \Cref{lem:rectangle},
    correctness forces
    $$
        x\in R\cup\{y\}
        \Longleftrightarrow
        x\in R'\cup\{y\}.
    $$
    Since $x\ne y$, this is equivalent to $x\in R\Longleftrightarrow x\in R'$.
    Hence $R\cap P_j(y)=R'\cap P_j(y)$ whenever the content of cell $j$ agrees.
    Thus, for fixed $y$ and $j$, the trace $R\cap P_j(y)$ is determined by the
    value $T_{R\cup\{y\}}[j]$.  This value lies in $Q$, so there are at most
    $m'$ possible traces.
\end{proof}

The following combinatorial lemma is the heart of the two-probe upper bound. It
says that an $(n-1)$-uniform family cannot be dense if, on each part of a
partition into at most $n$ pieces, it admits only $m'$ possible traces.

\begin{lemma}[Weighted trace count]
\label{lem:weightedtrace}
There is an absolute constant $\kappa>0$ such that the following holds for all
sufficiently large integers $n$. Let $U$ be a set of size $m'-1$, partitioned
into at most $n$ parts. Let
$$
        \mcalf\subseteq\binom{U}{n-1},
$$
and suppose that for every part $P$,
$$
        \left|\{F\cap P:F\in\mcalf\}\right|\le m'.
$$
If $m'\ge \kappa n(\log n)^2$, then
$$
        |\mcalf|\le n^{-3}\binom{m'-1}{n-1}.
$$
\end{lemma}

\begin{proof}
    Fix any constant $\kappa\ge1$; all estimates below are for sufficiently
    large $n$.  Put $u=m'-1$, $p=(n-1)/u$, and $\ell=u/n$. Sample a random set
    $B\subseteq U$ by including each element independently with probability
    $p$.

    We first refine the partition into pieces of controlled size. Split each original part $P$ into pieces of size
    between $\ell/2$ and $\ell$, except possibly for one leftover piece of size
    less than $\ell/2$. Call the non-leftover pieces good. The total size of
    the leftover pieces is at most $n\ell/2=u/2$, so the good pieces contain at
    least $u/2$ elements. Since each good piece has size at most $\ell$, there
    are at least $n/2$ good pieces.

    The trace bound survives refinement: if $H\subseteq P$ and $P$ is an
    original part, then
    $$
        \{F\cap H:F\in\mcalf\}
        =
        \{A\cap H:A\in\{F\cap P:F\in\mcalf\}\}.
    $$
    Thus every refined piece has at most $m'$ possible traces.

    Let $\Tr_H(\mcalf)=\{F\cap H:F\in\mcalf\}$.  If $B\in\mcalf$, then
    $B\cap H\in\Tr_H(\mcalf)$ for every refined piece $H$.  The random sets
    $B\cap H$ are independent over distinct pieces, so the corresponding events
    are independent.  Therefore
    \begin{equation}
        \label{eq:trace_independence}
        \Pr[B\in\mcalf]
        \le
        \prod_H \Pr[B\cap H\in\Tr_H(\mcalf)] .
    \end{equation}

    Fix a good piece $H$.  Since $\ell/2\le |H|\le\ell$ and
    $p=(n-1)/u$,
    $$
        \frac{n-1}{2n}
        \le
        |H|p
        \le
        \frac{n-1}{n}.
    $$
    Moreover, $p\le1/2$ for all sufficiently large $n$.  The inequality
    $\ln(1-p)\ge-2p$ then gives
    $$
        (1-p)^{|H|}
        \ge
        \exp(-2p|H|)
        \ge
        e^{-2}.
    $$
    Consequently, for every nonnegative integer $t$ with $t+1\le |H|$,
    $$
    \begin{aligned}
        \Pr[|B\cap H|=t+1]
        &=
        \binom{|H|}{t+1}p^{t+1}(1-p)^{|H|-t-1}\\
        &\ge
        \left(\frac{|H|p}{t+1}\right)^{t+1}(1-p)^{|H|}\\
        &\ge
        (c_*(t+1))^{-(t+1)}
    \end{aligned}
    $$
    for a suitable absolute constant $c_*>1$.

    Fix such a constant $c_*$, and for all sufficiently large $n$ let $s$ be
    the largest nonnegative integer satisfying
    $$
        (c_*(s+1))^{s+1}\le \frac{n}{100\log n}.
    $$
    Taking logarithms and using the maximality of $s$ gives
    $$
        s=\frac{\log n}{\log\log n-\log\log\log n+O(1)}.
    $$
    Since $m'\ge\kappa n(\log n)^2$, we have
    $\ell\ge(\kappa/2)(\log n)^2$ for all large $n$, and hence
    $$
        \frac{s^2}{\ell}
        =
        O\!\left(\frac{1}{(\log\log n)^2}\right)
        =o(1).
    $$
    In particular, $s+1\le |H|$ for all large $n$.  Applying the preceding
    point-probability estimate with $t=s$ gives
    $$
        \Pr[|B\cap H|=s+1]
        \ge
        (c_*(s+1))^{-(s+1)}
        \ge
        \frac{100\log n}{n}.
    $$
    Therefore
    $$
        \Pr[|B\cap H|\le s]\le 1-\frac{100\log n}{n}.
    $$

    We now split the allowed traces according to their size. Traces of size at
    most $s$ contribute at most the last probability. Each trace of size at
    least $s+1$ has probability at most $p^{s+1}$, and there are at most $m'$
    allowed traces. Hence
    $$
        \Pr[B\cap H\in\Tr_H(\mcalf)]
        \le
        1-\frac{100\log n}{n}+m'p^{s+1}.
    $$
    Since $m'p\le2n$ and
    $p\le 2/(\kappa(\log n)^2)$ for all large $n$,
    $$
        m'p^{s+1}
        \le
        2n\left(\frac{2}{\kappa(\log n)^2}\right)^s.
    $$
    Taking logarithms and using the displayed estimate for $s$, the logarithm
    of the right-hand side is
    $$
        1+\log n+s\bigl(\log(2/\kappa)-2\log\log n\bigr)
        =
        -\log n
        -\Omega\!\left(
            \frac{\log n\,\log\log\log n}{\log\log n}
        \right).
    $$
    Therefore
    $$
        m'p^{s+1}=o(1/n).
    $$

    Thus every good piece contributes the loss
    $$
        \Pr[B\cap H\in\Tr_H(\mcalf)]
        \le
        1-\frac{99\log n}{n}.
    $$

    By \Cref{eq:trace_independence}, multiplying over the at least $n/2$ good pieces gives
    $$
        \Pr[B\in\mcalf]
        \le
        \left(1-\frac{99\log n}{n}\right)^{n/2}
        \le n^{-40}
    $$
    for all sufficiently large $n$.

    Finally let $K_0$ be a uniformly random set from $\binom{U}{n-1}$. Since
    every member of $\mcalf$ has size exactly $n-1$,
    $$
        \Pr[B\in\mcalf]
        =
        \Pr[|B|=n-1]\Pr[K_0\in\mcalf].
    $$
    Write $\mu=up=n-1$. We need only the elementary central binomial-layer
    estimate, which follows directly from Stirling's formula: for an absolute
    constant $c_0>0$,
    $$
    \begin{aligned}
        \Pr[|B|=\mu]
        &=
        \binom{u}{\mu}p^\mu(1-p)^{u-\mu}  \\
        &\ge
        c_0
        \sqrt{\frac{u}{\mu(u-\mu)}}
        \left(\frac{u}{\mu}\right)^\mu
        \left(\frac{u}{u-\mu}\right)^{u-\mu}
        p^\mu(1-p)^{u-\mu}  \\
        &=
        \frac{c_0}{\sqrt{up(1-p)}}.
    \end{aligned}
    $$
    Since $up(1-p)=\Theta(n)$, this gives
    $$
        \Pr[|B|=n-1]\ge c n^{-1/2}
    $$
    for an absolute constant $c>0$. Consequently
    $$
        \Pr[K_0\in\mcalf]\le c^{-1} n^{1/2}n^{-40}\le n^{-3}
    $$
    for all large $n$, which is the desired bound.
\end{proof}

\begin{proof}[Proof of \Cref{thm:intro-g2}]
    Let $\kappa>0$ be the absolute constant from \Cref{lem:weightedtrace}.
    Assume, toward a contradiction, that a two-probe scheme exists for universe
    size
    $$
        m\ge \kappa n^2(\log n)^2 .
    $$
    Form the common first-probe class $Q$ as above. Then
    $m'=|Q|\ge m/n\ge \kappa n(\log n)^2$.

    For every $y\in Q$, apply \Cref{lem:weightedtrace} to
    $$
        \mcalf_y\subseteq\binom{Q\setminus\{y\}}{n-1}
    $$
    with the partition $P_1(y),\ldots,P_n(y)$. By \Cref{lem:second_trace},
    $$
        |\mcalf_y|\le n^{-3}\binom{m'-1}{n-1}.
    $$
    Summing over $y\in Q$ gives
    $$
        \sum_{y\in Q}|\mcalf_y|
        \le
        m'n^{-3}\binom{m'-1}{n-1}
        =
        n^{-2}\binom{m'}{n}
        <
        \binom{m'}{n}.
    $$
    The equality uses
    $m'\binom{m'-1}{n-1}=n\binom{m'}{n}$.
    This contradicts \Cref{eq:partition_f}, which counts every $n$-subset of
    $Q$ exactly once by its first-cell key and residual set. Therefore
    $m<\kappa n^2(\log n)^2$ for all sufficiently large $n$, which proves the
    claimed $O(n^2(\log n)^2)$ bound.
\end{proof}

\section{\texorpdfstring{The fixed-$q$ upper bound}{The fixed-q upper bound}}
\label{sec:general-q}

\paragraph{Yao's Ramsey upper bound.}
Yao's upper bound on the universe threshold comes from a Ramsey argument.
Given a storage scheme, color each set $S=\{s_1<\cdots<s_n\}$ by the permutation
that records the order in which the ranks $s_1,\ldots,s_n$ appear in the table. If
the universe is large enough, Ramsey's theorem gives a set $H$ of size $2n-1$
on which this coloring is constant. On $H$, every stored $n$-set has the same
rank order in the table, so the scheme behaves like a fixed permuted sorted
table. Yao's adversary argument for sorted tables then forces
$\lceil\log_2(n+1)\rceil$ probes. Equivalently,
$$
        G_q(n)\le R_{n!}^{(n)}(2n-1)-1<\Tow_n(n^{O(n)})
$$
whenever $q<\lceil\log_2(n+1)\rceil$.

\paragraph{Improved upper bounds on $G_q(n)$.}
The preceding argument colors whole $n$-sets, although a $q$-probe query sees
only $q$ table entries.  A crude coloring of the probe patterns on $(q+1)$-sets already gives an upper
bound of tower height $q+1$; see page~\pageref{pg:fixed-q_overview} for
discussion.

In this section, we prove the following quantitative form of \Cref{thm:intro-fixedq}: for every
fixed integer $q\ge3$, there is a constant $c_q$ such that, for all
sufficiently large $n$,
$$
        G_q(n)
        \le
        \Tow_{q-1}\!\left(
            n\exp\!\left(c_q\sqrt{\log n\,\log\log n}\right)
        \right).
$$
In particular,
$$
        G_q(n)\le \Tow_{q-1}\!\left(n^{1+o(1)}\right).
$$

\paragraph{Proof overview.}
Below we explain in more details how our proof to \Cref{thm:intro-fixedq} works.
The same as in the two-probe bound, we first restrict to a large set of queries that share the same first-probe
address. We next find a smaller universe on which the first
$q-1$ probes stay within a small fixed set of cells. To describe the coloring that achieves this, view the algorithm for
each query $x$ as a decision tree. Each internal node specifies a cell to
probe, and each outgoing edge specifies the key returned by that probe.
Thus a sequence of replies determines a node and the next cell to be read.

Partition the table cells into blocks, whose size we will choose below.
For a $(q-1)$-set $W$, consider the decision tree for every query $x\in W$.
For each ordered sequence of distinct replies from $W\setminus\{x\}$,
record the block containing the next cell probed. We do this only for
sequences of length $1,\ldots,q-2$, so the recorded nodes correspond to
probes $2,\ldots,q-1$. Collect all these block labels into a vector,
indexing its entries by the ranks in $W$ of the query and the ordered
replies. This vector is the \emph{color} of $W$.

For example, when $q=4$ and $W=\{a,b,c\}$, the tree for query $a$ contributes
four labels: the blocks probed after replies $b$, $c$, $(b,c)$, and $(c,b)$.
The trees for queries $b$ and $c$ each contribute four more. The color of
$W$ therefore has $12$ entries. In general, we write $a_q$ for the number
of entries; this number depends only on $q$.

Suppose Ramsey's theorem gives a universe $U$ on which all $(q-1)$-sets
have the same color. The common vector mentions at most $a_q$ blocks.
Every execution with query and stored keys in $U$ must use these blocks
for probes $2,\ldots,q-1$, as long as it has not yet read the query key.
Indeed, before any such probe, the query and the replies seen so far can
be included in a $(q-1)$-subset of $U$, and the next block is recorded in
that subset's color. Let $J$ consist of these blocks together with the
common first cell.

Now fix the contents of $J$. For every query outside the set of fixed
keys, the first $q-1$ probes and their replies are determined, so the
address of the last probe is determined as well. Partition these queries
by their last-probe address. On each part, the content of that one cell
determines all membership answers, giving the same trace bound as in
\Cref{sec:two-probe}. The counting lemma below then bounds the density
of each residual family. Larger blocks give fewer possible colors, but
also leave more cells in $J$ whose contents must be fixed. We choose the
block size so that the density bound is still small enough after summing
over all possible contents of $J$.

The counting step requires a version of \Cref{lem:weightedtrace} for
residual sets of size between $n/2$ and $n$, since we now fix several
cells. Its proof uses the same probabilistic argument and is deferred to
\Cref{app:parameterized_trace}.

\begin{lemma}[Trace count for large universes]
\label{lem:parameterized_trace}
There is an absolute constant $c>0$ such that the following holds for all
sufficiently large $n$. Put
$$
        \rho=\exp\!\left(c\sqrt{\log n\,\log\log n}\right).
$$
Let $U$ be a set of size $u\ge n\rho$, partitioned into at most $n$ parts, and let
$$
        \mcalf\subseteq\binom{U}{k},
        \qquad
        \frac n2\le k\le n.
$$
If, for every part $P$,
$$
        \left|\{F\cap P:F\in\mcalf\}\right|\le 2u,
$$
then
$$
        |\mcalf|
        \le
        \exp(-n/\rho)\binom{u}{k}.
$$
\end{lemma}

\begin{proof}[Proof of \Cref{thm:intro-fixedq}]
    Fix a constant $q\ge3$, and suppose a $q$-probe scheme exists on universe
    $[m]$. We may assume that the query algorithm makes exactly $q$ probes
    and never probes the same cell twice in one query. Since the first
    probe depends only on the query, there is a cell $j_*$ and a query
    class $C\subseteq[m]$ with $|C|\ge m/n$ such that every query in $C$
    first probes $j_*$.

    We first define the coloring for a partition of the table into $\nu$
    blocks. We will choose the block size after counting the entries in
    a color. Write $\operatorname{blk}(j)$ for the label of the block
    containing cell $j$, and fix a set
    $W=\{w_1<\cdots<w_{q-1}\}\subseteq C$.
    A node just before probe $t$ is specified by the query and the
    ordered sequence of the first $t-1$ replies. For each
    $2\le t\le q-1$, choose distinct indices
    $i_0,i_1,\ldots,i_{t-1}\in[q-1]$, representing the query $w_{i_0}$
    and the replies $w_{i_1},\ldots,w_{i_{t-1}}$. Let
    $\alpha_t(w_{i_0};w_{i_1},\ldots,w_{i_{t-1}})$ be the next cell
    probed after these replies; set $\alpha_t=1$ if this prefix cannot
    occur. The color of $W$ has one coordinate for each such choice,
    with value
    $$
        \operatorname{blk}\!\left(\alpha_t(w_{i_0};w_{i_1},\ldots,w_{i_{t-1}})\right).
    $$
    Order the coordinates by $t$ and then lexicographically by
    $(i_0,\ldots,i_{t-1})$. Thus two sets have the same color precisely
    when the block labels agree for every query rank and ordered
    sequence of reply ranks.

    There are $q-1$ choices for the query, $q-2$ for the first reply,
    and so on, down to $q-t$ for the last reply. Hence the number of
    coordinates recording probe $t$ is
    $(q-1)_t=(q-1)(q-2)\cdots(q-t)$, and the total number is
    $$
        a_q=\sum_{t=2}^{q-1}(q-1)_t.
    $$
    Since each coordinate is one of $\nu$ block labels, there are at
    most $\lambda=\nu^{a_q}$ colors. Each color mentions at most $a_q$
    distinct blocks.

    We now choose the universe size and the block size. Let $\rho$ be as
    in \Cref{lem:parameterized_trace}, and put
    $$
        m'=\lceil2n\rho\rceil.
    $$
    We will seek a homogeneous universe of size $m'$. This leaves room
    to remove the keys in the fixed cells while keeping at least
    $n\rho$ keys, as required by the lemma. Note that $\rho=n^{o(1)}$.

    If each block has at most $b$ cells, the blocks mentioned in one
    color together with the first cell contain at most $1+a_qb$ cells.
    As the final count below shows, summing over their contents costs
    a factor of at most $n$ per fixed cell, whereas the trace lemma
    gives a density bound of $\exp(-n/\rho)$. We therefore choose $b$
    so that $(1+a_qb)\ln n\le n/(100\rho)$. Partition the $n$ table cells
    into blocks of size
    $$
        b=\left\lfloor\frac{n}{200a_q\rho\log n}\right\rfloor,
    $$
    except possibly for one leftover block of smaller size.
    Since $\rho\log n=n^{o(1)}$ and $a_q$ is fixed, for all sufficiently
    large $n$,
    $$
        b\ge\frac{n}{400a_q\rho\log n}\ge1.
    $$
    Thus the number of blocks satisfies
    $$
        \nu\le n/b+1\le c_q\rho\log n,
    $$
    for a constant $c_q$ depending only on $q$, and
    $$
        \lambda=\nu^{a_q}\le(c_q\rho\log n)^{a_q}.
    $$

    We will show that $|C|<R_\lambda^{(q-1)}(m')$.
    Suppose otherwise. Ramsey's theorem gives a set $U\subseteq C$
    of size $m'$ on which the coloring is constant. Let $P$ be the
    union of the blocks whose labels appear in the common color, and put
    $$
        J=P\cup\{j_*\},
        \qquad
        r=|J|.
    $$
    Since at most $a_q$ block labels appear,
    $$
        r\le1+a_qb\le1+\frac{n}{200\rho\log n}.
    $$
    For all sufficiently large $n$, this gives
    $$
        r\le\frac n{10}
        \qquad\text{and}\qquad
        r\ln n\le\frac{n}{100\rho}.
    $$
    Indeed, $r\ln n\le\ln n+n/(200\rho)$, and
    $\ln n=o(n/\rho)$ since $\rho=n^{o(1)}$.

    To see what homogeneity gives, consider any query $x\in U$ on a
    table storing a set $S\in\binom{U}{n}$. Just before probe $t$, for
    $2\le t\le q-1$, suppose the query has not yet read $x$, and let
    $y_1,\ldots,y_{t-1}$ be its replies so far. These keys are distinct
    from each other and from $x$. Complete
    $\{x,y_1,\ldots,y_{t-1}\}$ to a $(q-1)$-subset $W$ of $U$.
    The block containing the next cell is a coordinate of the color
    of $W$, so it is one of the blocks in $P$. Together with the
    common first cell, this confines the first $q-1$ probes to $J$
    until the query key is read.

    We now condition on the contents of these cells, which turns this
    confinement into a completely determined prefix transcript for each fixed
    query. For each stored set $S\in\binom{U}{n}$, define its signature on $J$ by
    $$
        \sigma(S)=T_S|_J.
    $$
    For a fixed signature $\sigma$, let $A_\sigma$ be the set of keys appearing
    in the cells of $J$, so $|A_\sigma|=r$, and define the residual family
    $$
        \mathcal{F}_\sigma
        =
        \{S\setminus A_\sigma:S\in\binom{U}{n}
        \text{ and }\sigma(S)=\sigma\}.
    $$
    Then
    $$
        \mathcal{F}_\sigma
        \subseteq
        \binom{U\setminus A_\sigma}{n-r}.
    $$

    Fix a signature $\sigma$ and a query $x\in U\setminus A_\sigma$. Among
    all tables with signature $\sigma$, the first $q-1$ probes on query $x$ are
    determined by $x$ and $\sigma$. The first probe is $j_*\in J$, so its
    content is fixed by $\sigma$. Inductively, if fewer than $q-1$ probes have
    been made and the later probes have all fallen in $P\subseteq J$, then the
    observed keys are fixed elements of $A_\sigma$ and are distinct from $x$.
    By the homogeneity property above, the next probed cell also lies in $P$,
    and its content is again fixed by $\sigma$.

    Consequently, after conditioning on $\sigma$, the last probe is a fixed
    cell depending only on $x$ and $\sigma$. Let this cell be $f_\sigma(x)$, and
    partition $U\setminus A_\sigma$ into the last-probe parts
    $$
        H_j(\sigma)=\{x\in U\setminus A_\sigma:f_\sigma(x)=j\},
        \qquad j\in[n].
    $$

    The last probe gives the trace constraint. Take
    $R,R'\in\mathcal{F}_\sigma$, and write
    $$
        S=A_\sigma\cup R,
        \qquad
        S'=A_\sigma\cup R'.
    $$
    Suppose the corresponding tables have the same content in cell $j$:
    $T_S[j]=T_{S'}[j]$. For every query $x\in H_j(\sigma)$, the first $q-1$
    probes have the same fixed transcript in both tables, and the last probe is
    cell $j$ with the same content. The full transcripts are identical, so
    correctness gives
    $$
        x\in S\Longleftrightarrow x\in S'.
    $$
    Since $x\notin A_\sigma$, this is equivalent to
    $x\in R\Longleftrightarrow x\in R'$. Thus
    $R\cap H_j(\sigma)=R'\cap H_j(\sigma)$ whenever the content of cell $j$
    agrees. There are at most $m'$ possible contents, since all stored keys
    lie in $U$, and
    $|U\setminus A_\sigma|=m'-r\ge m'/2$, so
    $$
        \left|\{R\cap H_j(\sigma):R\in\mathcal{F}_\sigma\}\right|
        \le m'
        \le 2(m'-r).
    $$

    We apply \Cref{lem:parameterized_trace} to the ground set
    $U\setminus A_\sigma$, the family $\mathcal{F}_\sigma$, and the partition
    into the parts $H_j(\sigma)$. Put
    $$
        u=m'-r,
        \qquad
        k=n-r.
    $$
    Since $r\le n/10$ and $m'=\lceil2n\rho\rceil$, we have
    $$
        \frac n2\le k\le n,
        \qquad
        u\ge2n\rho-\frac n{10}\ge n\rho.
    $$
    The trace bound above is at most $2u$, so the lemma gives
    $$
        |\mathcal{F}_\sigma|
        \le
        \exp(-n/\rho)\binom{m'-r}{n-r}.
    $$

    It remains to sum over signatures. Every $S\in\binom{U}{n}$ has a unique
    realized signature $\sigma$, and, for each fixed $\sigma$, the map
    $S\mapsto S\setminus A_\sigma$ is a bijection from the stored sets with
    signature $\sigma$ to $\mathcal{F}_\sigma$. A signature is an injective map
    from the $r$ cells in $J$ to $U$, so there are at most
    $(m')_r=m'(m'-1)\cdots(m'-r+1)$ possible signatures. Consequently,
    $$
        \binom{m'}{n}
        =
        \sum_\sigma |\mathcal{F}_\sigma|
        \le
        (m')_r\exp(-n/\rho)
        \binom{m'-r}{n-r}.
    $$
    Since
    $$
        (m')_r\binom{m'-r}{n-r}
        =
        (n)_r\binom{m'}{n}
        \le
        n^r\binom{m'}{n},
    $$
    we obtain
    $$
        1
        \le
        \exp\!\left(-\frac n\rho+r\ln n\right)
        \le
        \exp\!\left(-\frac{99n}{100\rho}\right)
        <1,
    $$
    where the second inequality uses $r\ln n\le n/(100\rho)$. This is a
    contradiction. Therefore
    $$
        |C|<R_\lambda^{(q-1)}(m').
    $$

    Since $|C|\ge m/n$, the last inequality implies
    $$
        m<nR_\lambda^{(q-1)}(m').
    $$
    By \Cref{lem:fixed_uniformity_ramsey},
    $$
        R_\lambda^{(q-1)}(m')
        \le
        \Tow_{q-1}(c_qm'\lambda\log\lambda).
    $$
    To estimate the top of the tower, note that
    $$
        \log\lambda
        \le
        a_q\bigl(O_q(1)+\log \rho+\log\log n\bigr)
        =
        O_q\!\left(\sqrt{\log n\,\log\log n}\right).
    $$
    Hence
    $$
        \lambda\log\lambda
        \le
        \exp\!\left(O_q\!\left(\sqrt{\log n\,\log\log n}\right)\right).
    $$
    Since $m'=\lceil2n\rho\rceil$, increasing $c_q$ if necessary gives
    $$
        c_qm'\lambda\log\lambda
        \le
        z,
        \qquad
        z:=n\exp\!\left(c_q\sqrt{\log n\,\log\log n}\right).
    $$
    Since $z\ge n$ and $q-1\ge2$, we obtain
    $$
        m
        <
        n\Tow_{q-1}(z)
        \le
        \Tow_{q-1}(2z).
    $$
    Increasing $c_q$ once more absorbs the factor $2$ into $z$. This proves the
    stated bound, and the $n^{1+o(1)}$ form follows from
    $$
        \exp\!\left(c_q\sqrt{\log n\,\log\log n}\right)=n^{o(1)}.
    $$
\end{proof}

\section{Extracting search schemes from membership with polynomial loss}
\label{sec:membership-search}

This section proves a coarse equivalence between implicit membership and
implicit search.  Although membership is formally weaker, we show that a
membership scheme with a constant probe bound can be restricted to a
subuniverse on which every accepting computation reads $x$.  Thus, for each
fixed $q$, membership and implicit search are equivalent up to a polynomial
loss in universe size.

If $H_q(n)=\infty$, \Cref{thm:intro-membership-search} is immediate, so assume
throughout this section that $H_q(n)<\infty$.  The binary-search observation in
the preliminaries gives $q<\lceil\log_2(n+1)\rceil\le n$, and hence
$q\le n-1$.  Fix an arbitrary deterministic $q$-probe membership scheme on
universe $[m]$ that works for all $n$-sets.  Pad the query algorithm so that
every run makes exactly $q$ probes; after an early decision, the remaining
probes access previously unprobed cells and the final answer is unchanged.
Using the transcript notation from the preliminaries, write a full transcript
as
$$
        \tau=((\alpha_1,\beta_1),\ldots,(\alpha_q,\beta_q)).
$$
We only use transcripts realized by some query-input pair, so the addresses and
values appearing in $\tau$ are distinct.

\begin{definition}[Blind-accepting transcript]
For a query $x$, say that $\tau$ is \emph{blind-accepting for $x$} if,
when the query algorithm is run on $x$ and given the replies described by
$\tau$, it probes the addresses $\alpha_1,\ldots,\alpha_q$ in order, accepts,
and does not see $x$:
$$
        x\notin\{\beta_1,\ldots,\beta_q\}.
$$
Let
$$
        B(\tau)=
        \{x\in[m]\setminus\{\beta_1,\ldots,\beta_q\}:
        \tau\text{ is blind-accepting for }x\}.
$$
If $B(\tau)\ne\emptyset$, we call $\tau$ a blind-accepting transcript.
\end{definition}

\begin{definition}[Core of a blind-accepting transcript]
For a blind-accepting transcript $\tau$, fix one representative
$b(\tau)\in B(\tau)$ and define the core of $\tau$ to be
$$
        C(\tau)=\{\beta_1,\ldots,\beta_q,b(\tau)\},
$$
which has size $q+1$.
\end{definition}

Here is why we introduce the core.  Say that a table $T_S$ realizes $\tau$ if
$T_S[\alpha_i]=\beta_i$ for every $i\in[q]$.  If $x\in B(\tau)$ and $T_S$
realizes $\tau$, then, inductively over the probes, the actual run on $(x,S)$
follows $\tau$: it probes the prescribed addresses and receives the prescribed
replies.  This run accepts, so correctness forces $x\in S$.  Hence every stored
set whose table realizes $\tau$ contains
$\{\beta_1,\ldots,\beta_q\}\cup B(\tau)$, and in particular contains
$C(\tau)$.  The representative $b(\tau)$ merely selects one fixed $(q+1)$-set
that every table realizing $\tau$ must contain.

\begin{proof}[Proof of \Cref{thm:intro-membership-search}]
Put
$$
        \ell=H_q(n)+1.
$$
Since the unique $n$-set on universe $[n]$ can be stored in increasing order
and searched without probes, $H_q(n)\ge n$.  Together with $q\le n-1$, this
gives $q+1\le n<\ell$.
Suppose, toward a contradiction, that
\begin{equation}
    \label{eq:m_bound}
        m>n^q\ell^{q+1}.
\end{equation}
In particular, $m>\ell$.  We will find a subuniverse of size $\ell$ on which
the membership scheme becomes a $q$-probe implicit search scheme.  This
contradicts the definition of $H_q(n)$, since $\ell=H_q(n)+1$.

Let $\mathcal T$ be the set of realizable blind-accepting transcripts.  For
every $\tau\in\mathcal T$, form the core $C(\tau)$.  This entire family of
cores is fixed before the random choice below.

Choose a uniformly random $\ell$-subset $Y\subseteq[m]$, and let
$$
        \zeta=
        \bigl|\{\tau\in\mathcal T:C(\tau)\subseteq Y\}\bigr|.
$$
There are at most $n^qm^q$ choices of the addresses and values, so
$|\mathcal T|\le n^qm^q$.  If $C$ is a fixed core, then
$$
        \Pr[C\subseteq Y]
        =
        \frac{\binom{m-q-1}{\ell-q-1}}{\binom{m}{\ell}}
        =
        \prod_{i=0}^{q}\frac{\ell-i}{m-i}
        \le \left(\frac{\ell}{m}\right)^{q+1}.
$$
Therefore
$$
\begin{aligned}
\E[\zeta]
&=
\sum_{\tau\in\mathcal T}\Pr[C(\tau)\subseteq Y]  \\
&\le
n^qm^q\left(\frac{\ell}{m}\right)^{q+1}  \\
&=
\frac{n^q\ell^{q+1}}{m}.
\end{aligned}
$$
By \Cref{eq:m_bound}, $\E[\zeta]<1$.  Since $\zeta$ is integer-valued, some
$\ell$-set $Y$ has $\zeta=0$.  No stored set $S\subseteq Y$ can have a table
that realizes a blind-accepting transcript $\tau$: the observation above would
give $C(\tau)\subseteq S\subseteq Y$, contradicting $\zeta=0$.  Thus selecting
one representative for each transcript suffices to exclude every
query-transcript pair $(x,\tau)$ with $x\in B(\tau)$, even if $Y$ contains some
other such $x$.

Restrict the original membership scheme to stored sets $S\in\binom{Y}{n}$ and
queries $x\in Y$.  We turn it into a search scheme as follows.  Run the padded
membership query algorithm, and remember a cell if its content is $x$.  If the
membership algorithm rejects, output ``not present''.  If it accepts, output
the remembered cell.

For $x\notin S$, correctness of the membership scheme gives rejection.  For
$x\in S$, the membership algorithm accepts.  If it did not read a cell
containing $x$, then its padded accepting transcript $\tau$ would be
blind-accepting for $x$ and would be realized by the table of $S\subseteq Y$,
contrary to the preceding paragraph.  Thus whenever $x\in S$, the algorithm
reads $x$, and the remembered cell is a correct search answer.

This is the promised $q$-probe implicit search scheme on a universe of size
$\ell=H_q(n)+1$, a contradiction.  Thus the displayed inequality is impossible,
and hence
$$
        m\le n^q\bigl(H_q(n)+1\bigr)^{q+1}.
$$
Taking the maximum over admissible $m$ proves the theorem.
\end{proof}

\phantomsection
\addcontentsline{toc}{section}{Acknowledgements}
\section*{Acknowledgements}

The author would like to thank Vikrant Ashvinkumar for helpful discussions. The author would also like to thank Periklis A. Papakonstantinou for valuable suggestions in improving the presentation of this paper.

\paragraph{Statement on AI use.}
The proofs in this paper were first generated by ChatGPT 5.5 Pro.  The initial
prompt was broad: we asked for possible improvements to Yao's implicit
membership problem, without giving mathematical hints.  ChatGPT first produced
an exponential two-probe upper bound.  After several follow-up prompts,
especially prompts asking how to extend the argument to fixed $q$ and to refine
the $G_2(n)$ bound, it produced the near-quadratic two-probe proof and the
$\Tow_{q-1}(n^{1+o(1)})$ fixed-$q$ proof developed here.  In a separate
conversation, while we were looking for an improved bound for $G_4(n)$,
ChatGPT produced the equivalence between implicit membership and implicit
search.  The authors have validated and edited the proofs.  The authors assume
responsibility for all content.

\bibliographystyle{alpha}
\bibliography{ref}

\appendix

\section{Proof of the trace count for large universes}
\label{app:parameterized_trace}

\begin{proof}[Proof of \Cref{lem:parameterized_trace}]
    We first choose the cutoff used in the counting argument. Fix a
    sufficiently large absolute constant $d\ge4e$, and put
    $$
        s=\left\lfloor\sqrt{\frac{\log n}{\log\log n}}\right\rfloor,
        \qquad
        b_s=d^s(s+1)!,
        \qquad
        \ell=\frac un.
    $$
    Stirling's formula gives
    $$
        \log b_s=O(s\log s)
        =O\!\left(\sqrt{\log n\,\log\log n}\right)
        =o(\log n),
    $$
    and
    $$
        \log\bigl((4nb_s)^{1/s}\bigr)
        =\frac{\log n+\log b_s+2}{s}
        =O\!\left(\sqrt{\log n\,\log\log n}\right).
    $$
    Thus, choosing the absolute constant $c$ in the lemma sufficiently large
    ensures that, for all sufficiently large $n$,
    $$
        \rho\ge(4nb_s)^{1/s}
        \qquad\text{and}\qquad
        \rho\ge8b_s.
    $$
    Since $u\ge n\rho$, we have $\ell\ge \rho$ and hence $\ell^s\ge4nb_s$.

    Put $p=k/u$, so $p\le1/\rho\le1/2$ for all sufficiently large $n$, and sample
    a random set $B\subseteq U$ by including each element independently with
    probability $p$. We first bound $\Pr[B\in\mcalf]$ under this product
    distribution, and then return to the uniform layer $\binom{U}{k}$.

    We begin by regularizing the parts. Split every original part into pieces
    of size between $\ell/2$ and $\ell$, except possibly for one leftover piece
    of size less than $\ell/2$. Call the non-leftover pieces good. Since there
    are at most $n$ original parts, the total size of the leftovers is at most
    $n\ell/2=u/2$. Hence the good pieces contain at least $u/2$ elements, and
    since each good piece has size at most $\ell$, there are at least $n/2$
    good pieces.

    The trace bound survives this refinement. If $H\subseteq P$, then
    $$
        \{F\cap H:F\in\mcalf\}
        =
        \{A\cap H:A\in\{F\cap P:F\in\mcalf\}\}.
    $$
    Thus every refined piece has at most $2u$ possible traces. Write
    $$
        \mcala_H=\{F\cap H:F\in\mcalf\}.
    $$

    If $B\in\mcalf$, then $B\cap H\in\mcala_H$ for every refined piece $H$.
    The random sets $B\cap H$ are independent over different pieces, so
    $$
        \Pr[B\in\mcalf]
        \le
        \prod_H \Pr[B\cap H\in\mcala_H].
    $$

    Fix a good piece $H$, and write $a=|H|$. Then $\ell/2\le a\le \ell$.
    Since $n/2\le k\le n$, the mean of $|B\cap H|$ is bounded between two
    constants:
    $$
        \frac14\le ap\le 1.
    $$
    We split the allowed traces according to their size. Traces of size at most
    $s$ contribute at most $\Pr[|B\cap H|\le s]$. Each trace of size at least
    $s+1$ has probability at most $p^{s+1}$, and there are at most $2u$ such
    traces. Hence
    $$
        \Pr[B\cap H\in\mcala_H]
        \le
        \Pr[|B\cap H|\le s]+2up^{s+1}.
    $$

    We now lower-bound the binomial layer $|B\cap H|=s+1$. The bound
    $\ell^s\ge4nb_s$, together with $s!\ge(s/e)^s$, gives
    $$
        \ell
        \ge d\bigl((s+1)!\bigr)^{1/s}
        \ge \frac{ds}{e}
        \ge 4s.
    $$
    Since $a\ge\ell/2$ and $s\ge1$, we have $a\ge2s\ge s+1$. Thus
    $$
        \binom{a}{s+1}
        \ge
        \left(\frac{a}{s+1}\right)^{s+1}
    $$
    holds. Also, $p\le1/2$ and $ap\le1$, so
    $\ln(1-p)\ge-2p$ gives $(1-p)^a\ge e^{-2}$. Therefore
    $$
        \Pr[|B\cap H|=s+1]
        =
        \binom{a}{s+1}p^{s+1}(1-p)^{a-s-1}
        \ge
        \frac{e^{-2}}{(4(s+1))^{s+1}}.
    $$
    By the standard bound $(s+1)!\ge((s+1)/e)^{s+1}$, we may choose the
    absolute constant $d$ large enough that
    $$
        \frac{e^{-2}}{(4(s+1))^{s+1}}
        \ge
        \frac{1}{d^s(s+1)!}
        =
        \frac1{b_s}
    $$
    for every $s\ge1$. Therefore
    $$
        \Pr[|B\cap H|\le s]\le 1-\frac1{b_s}.
    $$

    The contribution of the larger traces is small because $p\le 1/\ell$:
    $$
        2up^{s+1}
        \le
        2n\ell\left(\frac1\ell\right)^{s+1}
        =
        \frac{2n}{\ell^s}
        \le
        \frac1{2b_s}.
    $$
    Thus every good piece satisfies
    $$
        \Pr[B\cap H\in\mcala_H]\le 1-\frac1{2b_s}.
    $$

    Multiplying over the at least $n/2$ good pieces gives
    $$
        \Pr[B\in\mcalf]
        \le
        \left(1-\frac1{2b_s}\right)^{n/2}
        \le
        \exp\!\left(-\frac{n}{4b_s}\right).
    $$

    Finally let $K$ be a uniformly random set from $\binom{U}{k}$. Since every
    member of $\mcalf$ has size exactly $k$,
    $$
        \Pr[B\in\mcalf]
        =
        \Pr[|B|=k]\Pr[K\in\mcalf].
    $$
    Here $|B|$ is binomial with mean $k$, and
    $$
        \operatorname{Var}(|B|)
        =
        up(1-p)
        =
        k\left(1-\frac{k}{u}\right)
        \in\left[\frac n4,n\right].
    $$
    Thus the same estimate used in \Cref{lem:weightedtrace} gives
    $$
        \Pr[|B|=k]\ge c_*n^{-1/2}
    $$
    for an absolute constant $c_*>0$. Hence
    $$
        \Pr[K\in\mcalf]
        \le
        c_*^{-1}\sqrt n\exp\!\left(-\frac{n}{4b_s}\right).
    $$
    Since $K$ is uniform on $\binom{U}{k}$ and $b_s=n^{o(1)}$, we obtain,
    for all sufficiently large $n$,
    $$
        |\mcalf|
        \le
        \exp\!\left(-\frac{n}{8b_s}\right)\binom{u}{k}
        \le
        \exp(-n/\rho)\binom{u}{k}.
    $$
    The first inequality uses $\ln(c_*^{-1}\sqrt n)=O(\log n)=o(n/b_s)$
    to absorb the factor $c_*^{-1}\sqrt n$ into the exponent, and the second
    uses $\rho\ge8b_s$. This proves the lemma.
\end{proof}

\end{document}